\documentclass[10pt,twocolumn]{article}

\usepackage[T1]{fontenc}
\usepackage[utf8]{inputenc}
\usepackage[a4paper, top=2.0cm, bottom=2.4cm, left=1.5cm, right=1.5cm, columnsep=0.65cm]{geometry}
\usepackage{amsmath}
\usepackage{amssymb}
\usepackage{amsthm}
\usepackage{booktabs}
\usepackage{hyperref}
\usepackage{microtype}
\usepackage{graphicx}
\usepackage{url}
\usepackage{xcolor}
\usepackage{array}
\usepackage{tabularx}
\usepackage{titlesec}
\usepackage{caption}
\usepackage{enumitem}
\usepackage{natbib}
\usepackage{algorithm}
\usepackage{algpseudocode}
\usepackage{multirow}
\usepackage{fancyhdr}

\hypersetup{colorlinks=true,linkcolor=blue!60!black,citecolor=blue!60!black,urlcolor=blue!60!black}

\titleformat{\section}{\normalfont\large\bfseries\scshape}{\thesection}{1em}{}
\titleformat{\subsection}{\normalfont\normalsize\bfseries}{\thesubsection}{1em}{}
\titleformat{\subsubsection}{\normalfont\normalsize\itbfseries}{\thesubsubsection}{1em}{}

\setlist[itemize]{leftmargin=*, topsep=2pt, itemsep=1pt}
\setlist[enumerate]{leftmargin=*, topsep=2pt, itemsep=1pt}

\newcolumntype{Y}{>{\raggedright\arraybackslash}X}

\newtheorem{invariant}{Invariant}
\newtheorem{proposition}{Proposition}

\title{\textbf{bioETH-PRS: Confidential Polygenic Risk Scoring without\\[2pt]
       a Trusted Evaluator via Fully Homomorphic Encryption on\\[2pt]
       a Programmable Blockchain}}

\author{%
  \begin{minipage}{\textwidth}
  \raggedright
  \textbf{Kimon Antonios Provatas$^{1,+,*}$, Christos Galanopoulos$^{1,+}$, Ilias Georgakopoulos-Soares$^{1,,*}$}\\[2pt]
  \normalsize $^{1}$Division of Pharmacology and Toxicology, College of Pharmacy, The University of Texas at Austin, Dell Pediatric Research Institute, Austin, TX, USA\\[2pt]
  \small $^{+}$Co-first authors \quad $^{*}$Corresponding authors: \texttt{kap4722@my.utexas.edu}; \texttt{ilias@austin.utexas.edu}\\[2pt]
  \small\textit{Blockchain $\cdot$ Genomic Privacy $\cdot$ FHE $\cdot$ Smart Contracts $\cdot$ Privacy-Preserving Output Release}
  \end{minipage}
}
\date{}

\begin{document}
\maketitle
\thispagestyle{fancy}

\begin{figure*}[t]
  \centering
  \includegraphics[width=\textwidth, keepaspectratio]{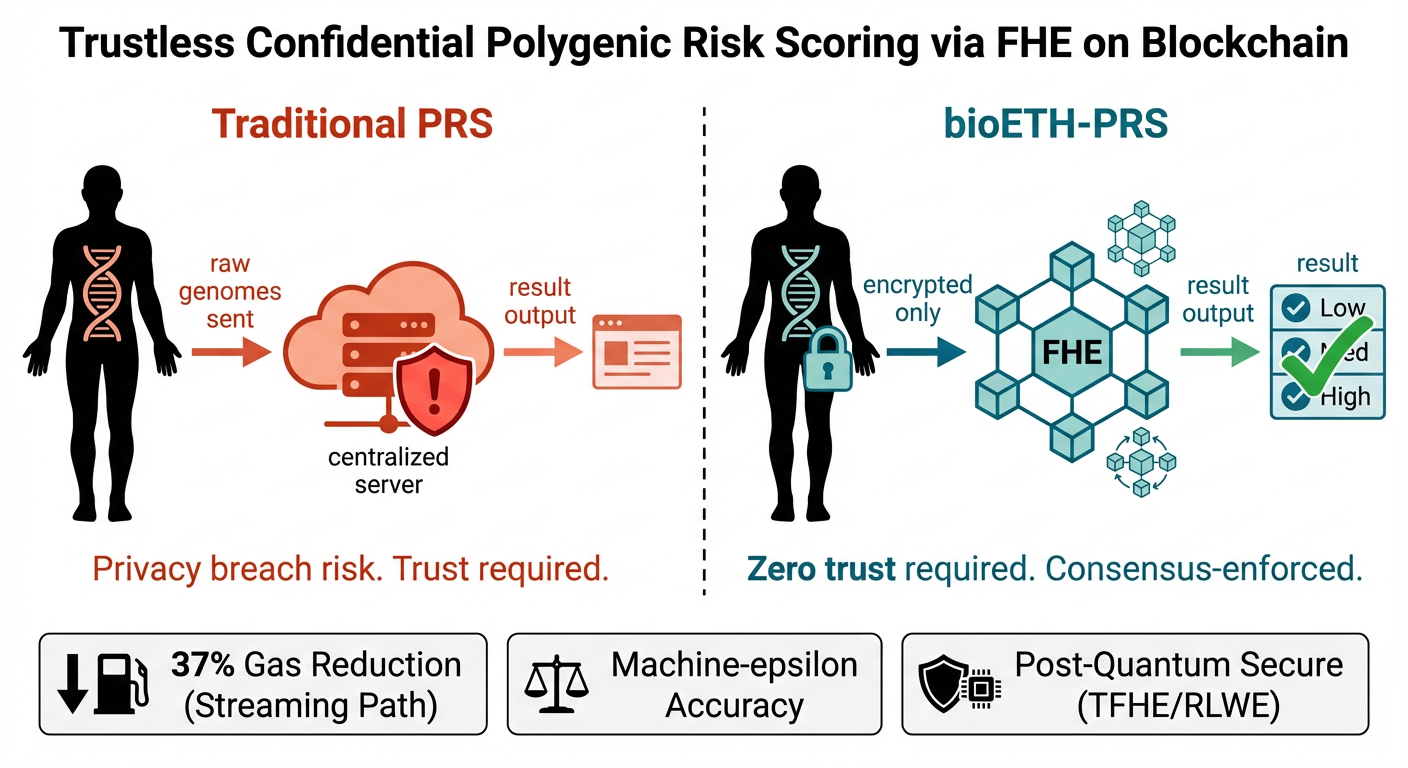}
  \caption{\textbf{Graphical Abstract.}  bioETH-PRS replaces the trusted evaluator of conventional
  homomorphic PRS pipelines with consensus-enforced smart contracts.
  \emph{Left:} Traditional centralised approaches transmit raw genotypes to a cloud evaluator,
  requiring user trust in a designated third party.
  \emph{Right:} bioETH-PRS keeps all genomic data and model weights encrypted throughout,
  with computation verified by blockchain consensus.
  Key results: 37\% gas reduction via streaming execution; machine-epsilon quantisation accuracy
  on real GWAS fixtures; post-quantum security under TFHE/RLWE.}
  \label{fig:graphical_abstract}
\end{figure*}

\begin{abstract}
\noindent
Polygenic risk scores (PRSs) aggregate genetic effect estimates to predict disease
susceptibility, yet clinical deployment often exposes raw genotype data to third-party
compute infrastructure.  Prior homomorphic-encryption approaches, still require trust in a designated evaluator.  We present \textbf{bioETH-PRS},
a protocol that replaces that evaluator role with immutable smart
contracts on a blockchain supporting Fully Homomorphic Encryption (fhEVM).  Using the
integer-exact TFHE scheme, bioETH-PRS computes the PRS dot product entirely within the
encrypted domain, keeping both genotype dosage vectors and GWAS weight vectors hidden
from external parties throughout execution.  We introduce a three-step fixed-point
quantisation scheme for representing signed GWAS weights as unsigned 64-bit integers,
achieving machine-epsilon reconstruction accuracy on validated fixtures. A four-contract
architecture separates data custody, model publication, computation, and output release,
and supports both a classic chunked path and a streaming path, with the latter reducing
mock-measured gas by 37\%.  An on-chain noisy output oracle emits an encrypted noisy-score handle and a publicly
decryptable ternary category, reducing raw score exposure and probing risk.  Prototype
evaluation on real GWAS fixtures confirms linear gas scaling and suggests that the
approach may be cost-competitive in low-gas deployment environments.
\end{abstract}

\section*{Key Points}
\begin{itemize}
  \item bioETH-PRS replaces the trusted evaluator in encrypted PRS computation with auditable smart contracts on an FHE-enabled blockchain.
  \item The protocol protects both patient genotypes and GWAS model weights through encrypted execution and ACL-gated output release.
  \item A fixed-point quantisation scheme enables exact TFHE-compatible PRS computation with machine-epsilon reconstruction accuracy on validated fixtures.
  \item A streaming execution path reduces mock-measured gas consumption by 37\% relative to the classic chunked path.
  \item Anti-probing controls combine noisy output release with per-model, per-wallet, and per-sample rate limiting.
\end{itemize}

\section{Introduction}
\label{sec:intro}

Advances in genome sequencing have made polygenic risk scores (PRSs) a central tool
in precision medicine \citep{wray2021polygenic,inouye2018genomic}.  A PRS aggregates
the dosage-weighted effects of genetic variants:
\begin{equation}
  \mathrm{PRS} = \sum_{i=1}^{N} g_i \cdot \beta_i,
  \label{eq:prs}
\end{equation}
where $g_i \in \{0,1,2\}$ is the allele dosage of a genetic variant $i$ and $\beta_i$ is the
corresponding GWAS effect weight.  Published polygenic scores in the PGS Catalog range
from small panels to genome-wide scores with large variant sets \citep{lambert2019polygenic},
and hospital deployment increasingly involves
transmitting raw genotype data to external compute infrastructure, a practice that
raises acute privacy concerns given the re-identification risk of genomic sequences
\citep{gymrek2013,erlich2014}.

\textbf{The double-privacy problem.}
Two classes of sensitive information require simultaneous protection: (1)
\emph{patient genotypes} ($g_i$), identifying, heritable, and permanent; and
(2) \emph{GWAS model weights} ($\beta_i$), derived from restricted research cohorts.
More generally, model interfaces and outputs can leak sensitive information through
model-inversion attacks \citep{fredrikson2015modelinversion}.
Centralised cloud solutions protect neither of them, and both genotypes and weights must
be decrypted before arithmetic can proceed.

\textbf{The FHE opportunity.}
Fully Homomorphic Encryption (FHE) permits arbitrary arithmetic on encrypted data
without decryption.  Knight et al.\ \citep{knight2026heprs} demonstrate this with
HEPRS, an open-source pipeline that computes a 110,000-SNP schizophrenia PRS under
the CKKS scheme \citep{cheon2017ckks} with near-zero accuracy loss ($r > 0.999$,
MSE~$< 2.3 \times 10^{-6}$).  HEPRS, however, remains architecturally centralised:
it relies on a three-party trust model (client, modeler, evaluator), and a colluding
evaluator and client can reveal GWAS model weights.  Trust in the evaluator is an
assumption, not a protocol guarantee.  Furthermore, HEPRS requires 3--4~GB of RAM
per individual even at 110,000 SNPs, and output is a raw score---there is no built-in
mechanism to prevent score accumulation or model probing.

\textbf{Our contribution.}
We propose \textbf{bioETH-PRS}, which removes the designated evaluator of prior
centralised FHE pipelines.  On a blockchain augmented with an FHE coprocessor
(fhEVM \citep{zamafhevm}), the evaluator is replaced by a deterministic, publicly
auditable smart contract.  This shifts trust from a single evaluator to auditable
contract logic, blockchain consensus, the fhEVM coprocessor stack, and the ACL/
decryption infrastructure (Figure~\ref{fig:graphical_abstract}).

\smallskip\noindent\textbf{Distinctions from HEPRS.}
Although both systems compute the same encrypted PRS inner product, bioETH-PRS adds
four design elements that are absent from HEPRS.  First, TFHE/fhEVM exposes only
unsigned integer arithmetic, requiring a dedicated fixed-point encoding and recovery
scheme for signed GWAS weights; HEPRS instead uses CKKS with native approximate
floating-point support.  Second, bioETH-PRS separates data custody, model publication,
computation, and output release into four independently auditable contracts with
different trust and upgrade surfaces.  Third, it introduces two fhEVM-specific
execution strategies, classic and streaming, that trade relay flexibility against
gas and storage overhead under the handle/SSTORE model.  Fourth, it restricts outputs
through an oracle-enforced post-processing layer with anti-probing rate limits, rather
than releasing raw scores directly.

\smallskip\noindent Our main contributions are:
\begin{itemize}
  \item A \textbf{four-contract architecture without a trusted evaluator} separating genomic data
        custody, model publication, computation, and post-processed output release---each independently
        auditable with distinct trust surfaces and upgrade lifecycles.
  \item A \textbf{three-step fixed-point quantisation scheme} bridging signed GWAS
        floats and unsigned TFHE integers with provable overflow safety and
        machine-epsilon reconstruction error, together with an automated advisor that
        selects the minimum viable scale factor before model publication.
  \item Two \textbf{gas-optimised execution strategies}: a classic chunked path
        enabling multi-party and relay-mediated computation, and a streaming path with
        37\% gas reduction for single-requester flows.
  \item An \textbf{on-chain noisy output oracle} with cryptographically generated noise,
        mandatory oracle-required enforcement, and minimum threshold-gap constraints
        that together reduce score probing risk and prevent raw-score bypass.
  \item A \textbf{rate-limiting mechanism} based on per-model, per-wallet, and per-sample,
        block-windowed job quotas that raise the cost of model weight extraction to
        thousands of hours at recommended settings.
  \item Prototype evaluation on real HEPRS GWAS fixtures (100--5,000 SNPs)
        confirming linear gas scaling, exact score correctness, and economically
        plausible analysis costs under L2-equivalent gas pricing assumptions.
\end{itemize}

\section{Background}
\label{sec:background}

\subsection{Polygenic Risk Scores}

A PRS is the weighted inner product of a patient's genotype dosage vector
$\mathbf{g} \in \{0,1,2\}^N$ and a GWAS effect weight vector $\boldsymbol{\beta}
\in \mathbb{R}^N$ (Eq.~\ref{eq:prs}).  Effect weights are typically estimated by
large-scale GWAS over tens of thousands of individuals and are signed floating-point
values in the range $[-0.5, +0.5]$.  Standard pipelines such as PLINK~2
\citep{chang2015plink} and LDpred2 \citep{prive2020ldpred2} compute PRSs in
milliseconds on plaintext data; FHE approaches trade runtime for privacy.

PRS models span a wide range in SNP inclusion depending on construction methodology. Sparse or clinically oriented models may include hundreds to a few thousand variants, while genome-wide approaches incorporate tens of thousands to millions of SNPs. In this work, we perform a proof of principle study with 5,000 SNPs, with on-chain FHE computation.

The computation in Eq.~\ref{eq:prs} is an inner product between two vectors of
moderate length.  This is exactly the primitive that FHE systems are designed to
support efficiently, making PRS computation a natural target for FHE-based privacy
enhancement.

\subsection{Fully Homomorphic Encryption}

FHE schemes differ in the arithmetic they support and their performance.  Two are
directly relevant to this work.

\textbf{CKKS} \citep{cheon2017ckks} supports approximate arithmetic over real-valued data, is
naturally suited to floating-point computations, and underpins HEPRS.  It allows
arbitrary depth computation at the cost of accumulating approximation error.  The
Lattigo library \citep{lattigo} implements CKKS in Go with the parameter sets used
by HEPRS.  CKKS is computationally expensive: evaluating a 110,000-SNP PRS on a
single CPU node requires approximately 4.9~s per individual and up to 130~GB of RAM
for a cohort of 1,000 \citep{knight2026heprs}.

\textbf{TFHE} \citep{chillotti2020tfhe} operates on exact unsigned integers and
supports bootstrapping to arbitrary computation depth.  Zama's fhEVM implements TFHE
as EVM precompile operations exposed to smart contracts.  TFHE provides deterministic
integer arithmetic, a better fit for smart contracts where bit-exact reproducibility
is required for consensus. The cost is that signed floating-point weights must be
converted to unsigned integer encodings, motivating our quantisation design
(Section~\ref{sec:quantization}).

TFHE security is grounded in the Ring Learning with Errors (RLWE) hardness assumption
\citep{chillotti2020tfhe}, which is believed to be quantum-resistant.  CKKS under the
RLWE assumption likewise provides post-quantum security.  Both are therefore
appropriate choices for long-term genomic data protection.

\subsection{Programmable Blockchain and fhEVM}

The Fully Homomorphic Ethereum Virtual Machine (fhEVM) extends the Ethereum
execution environment with TFHE precompile operations \citep{zamafhevm}.  Every
encrypted value is represented on-chain by a 32-byte opaque handle; actual
ciphertexts are managed by an off-chain coprocessor.  Smart contracts store only
handles and trigger FHE operations through coprocessor calls.

Access to encrypted outputs is governed by an on-chain ACL contract: a handle is
decryptable only by addresses explicitly granted access by the ACL.  The blockchain
is the trust anchor---no party can obtain a decryption unless the executing smart
contract grants it, and the smart contract's grant conditions are publicly auditable.

A per-transaction Homomorphic Computation Unit (HCU) budget limits the number of FHE
operations per block.  In our Ethereum mock-coprocessor prototype, we measured a
budget of 60--74 FHE operations per transaction (Section~\ref{sec:evaluation}); the
Sepolia testnet ceiling is expected to be higher but remains to be measured.  The
HCU constraint necessitates chunked computation
strategies for large input vectors.

\section{System Design}
\label{sec:design}

\subsection{Architecture Overview}

bioETH-PRS is structured as four independently deployable smart contracts forming
a linear pipeline from sample data custody to risk output (Figure~\ref{fig:arch}).
The four-layer separation reflects distinct trust surfaces and upgrade lifecycles:
the data registry protects patient data; the marketplace protects researcher GWAS
models; the compute engine implements computation logic; the oracle enforces privacy
policy.  Each contract is independently auditable and can be replaced without
affecting the others.

\begin{figure}[t]
  \centering
  \includegraphics[width=\columnwidth, keepaspectratio]{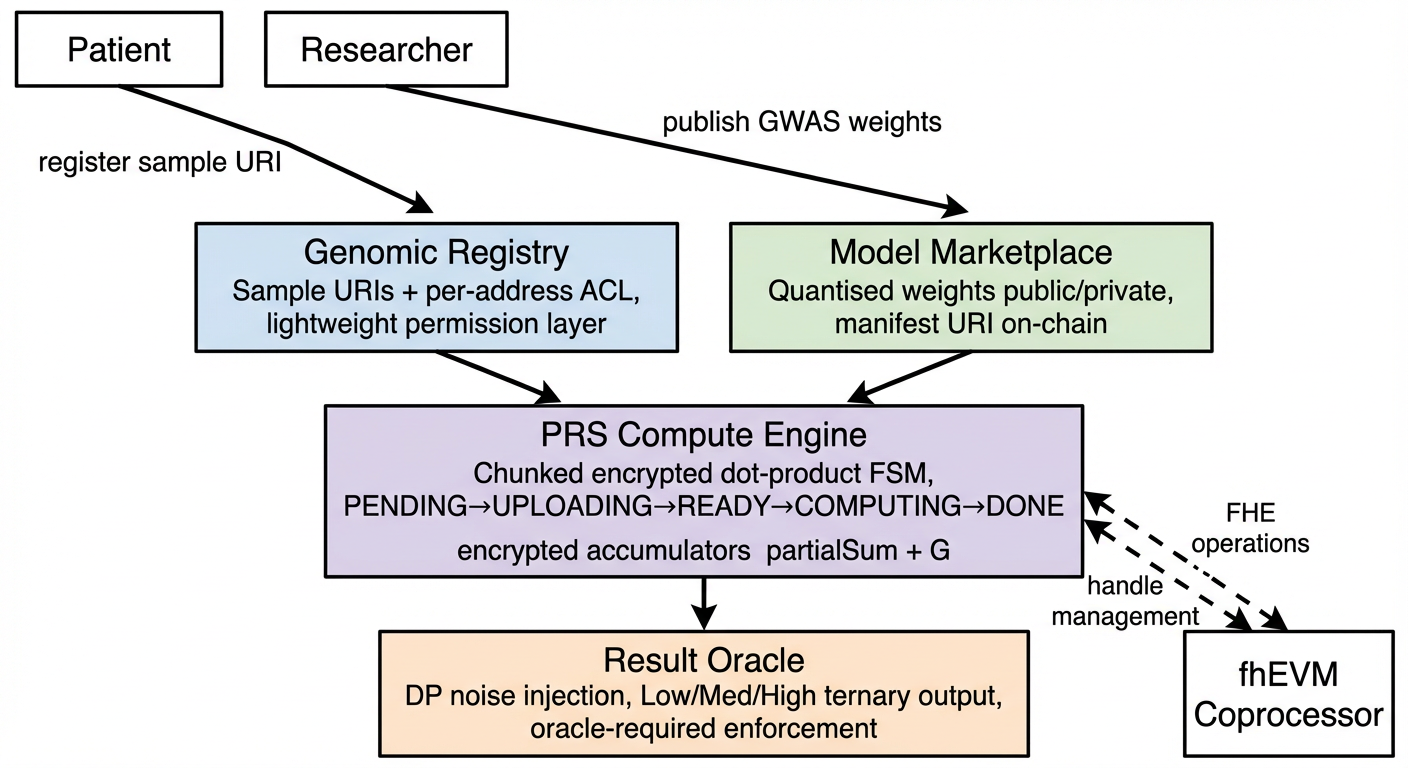}
  \caption{System architecture.  Four on-chain smart contracts form a linear pipeline
  from genomic data custody to privacy-preserving risk output.  The fhEVM coprocessor
  manages off-chain ciphertexts; only opaque 32-byte handles are stored on-chain.
  No plaintext genotype data or model weights are observable at any stage.}
  \label{fig:arch}
\end{figure}

\textbf{Genomic Registry.}
URI-based registry of sample datasets with per-address access control lists (ACLs).
Stores references (IPFS or Arweave URIs) to off-chain encrypted genotype data.
A patient registers samples and explicitly grants compute access to specific callers.
This contract performs no FHE operations; it is a lightweight permission layer.

\textbf{Model Marketplace.}
Immutable repository of GWAS weight vectors published in chunks.  Weights may be
stored as plaintext integers (public models, enabling a cheaper coprocessor-optimised
ciphertext-times-plaintext multiply, $\approx$60\% less expensive than
ciphertext-ciphertext) or as encrypted handles (private models).  A manifest URI
records quantisation metadata on-chain: scale factor, weight zero-point, score
offset, and provenance hash.  Models are immutable after finalisation; version
management is the responsibility of the publishing researcher.

\textbf{PRS Compute Engine.}
The core computation contract.  Implements a finite state machine
($\texttt{PENDING}\!\to\!\texttt{UPLOADING}\!\to\!\texttt{READY}\!\to\!\texttt{COMPUTING}\!\to\!\texttt{DONE}$)
that orchestrates chunked dot-product accumulation of patient SNPs against model
weights.  The engine tracks two encrypted running accumulators: the weighted partial
sum and the genotype sum (required for quantisation correction; see
Section~\ref{sec:quantization}).  Final score handles are ACL-granted only to the
requesting party.  Two parallel execution paths are supported (Section~\ref{sec:execution}).

\textbf{Result Oracle.}
Applies on-chain cryptographic noise and emits an encrypted noisy score handle plus
a publicly decryptable ternary risk category.  Noise is generated by the fhEVM's on-chain
random source, making it unknowable to any party before the transaction is mined.  The
oracle can be set as mandatory by the model owner, preventing direct requester access to
the raw score when oracle-required mode is enabled.

\subsection{Comparison with HEPRS}

Table~\ref{tab:comparison} summarises the key architectural differences between
bioETH-PRS and HEPRS \citep{knight2026heprs}.  HEPRS achieves near-perfect accuracy
on a 110,000-SNP schizophrenia model running on commodity hardware, but its
three-party trust assumption means a colluding evaluator and client can reveal model
weights.  bioETH-PRS removes the designated evaluator assumption by replacing that
role with auditable on-chain code.

The core mathematical operation is identical in both systems: the PRS dot product
of Eq.~\ref{eq:prs}.  The critical difference is in the trust model and the FHE
scheme.  CKKS handles signed floats natively; TFHE requires unsigned integer
encoding.  CKKS evaluates approximate arithmetic; TFHE arithmetic is exact.  The
evaluator-free architecture of bioETH-PRS comes at the cost of restricting the practical SNP
range to 100--5,000---a regime that captures a subset of curated PRS models
\citep{lambert2019polygenic} while remaining feasible for on-chain
computation.

\begin{table}[ht]
\centering
\caption{Architectural comparison: HEPRS vs.\ bioETH-PRS}
\label{tab:comparison}
\footnotesize
\setlength{\tabcolsep}{4pt}
\begin{tabularx}{\columnwidth}{@{}p{1.9cm}YY@{}}
\toprule
\textbf{Property} & \textbf{HEPRS} \citep{knight2026heprs} & \textbf{bioETH-PRS} \\
\midrule
FHE scheme        & CKKS (approx.\ floats)  & TFHE (exact integers)  \\
Evaluator         & Centralised server      & Smart contract         \\
Trust model       & 3-party non-collusion   & Evaluator removed; consensus/fhEVM dependent  \\
Weight privacy    & Evaluator-gated         & ACL-gated, immutable   \\
Output control    & External (optional)     & On-chain noisy release \\
GWAS weight type  & Signed floats           & Quant.\ uint64         \\
Max variants tested   & 110,000                 & 5,000 (scalable)       \\
Per-person latency& $\sim$4.9 s            & $\sim$386 ms (100 SNPs)\\
Score access      & By client               & Noisy score + category \\
Post-quantum      & Yes (CKKS/RLWE)         & Yes (TFHE/RLWE)        \\
\bottomrule
\end{tabularx}
\end{table}

\section{Quantisation Scheme}
\label{sec:quantization}

\subsection{The Representation Problem}

TFHE arithmetic on the fhEVM operates on unsigned 64-bit integers.
GWAS weights $\beta_i$ are signed floating-point values;
a naive fixed-point encoding that multiplies by a scale factor $s$ and rounds
produces quantised weights $q_i = \mathrm{round}(s \cdot \beta_i) \in \mathbb{Z}$,
which may be negative.  Negative integers cannot be stored in an unsigned 64-bit
integer and would silently wrap under modular overflow, corrupting results without any
observable error.  A correct encoding must guarantee that all intermediate encrypted
values remain in the non-negative range $[0, 2^{64} - 1]$ for any genomically valid
input $g_i \in \{0, 1, 2\}$.

\subsection{Three-Step Unsigned Encoding}

We address this with a three-step bijection that maps the space of signed
floating-point dot products into a guaranteed non-negative 64-bit integer range
(Figure~\ref{fig:quantization}).

\begin{figure}[t]
  \centering
  \includegraphics[width=\columnwidth, keepaspectratio]{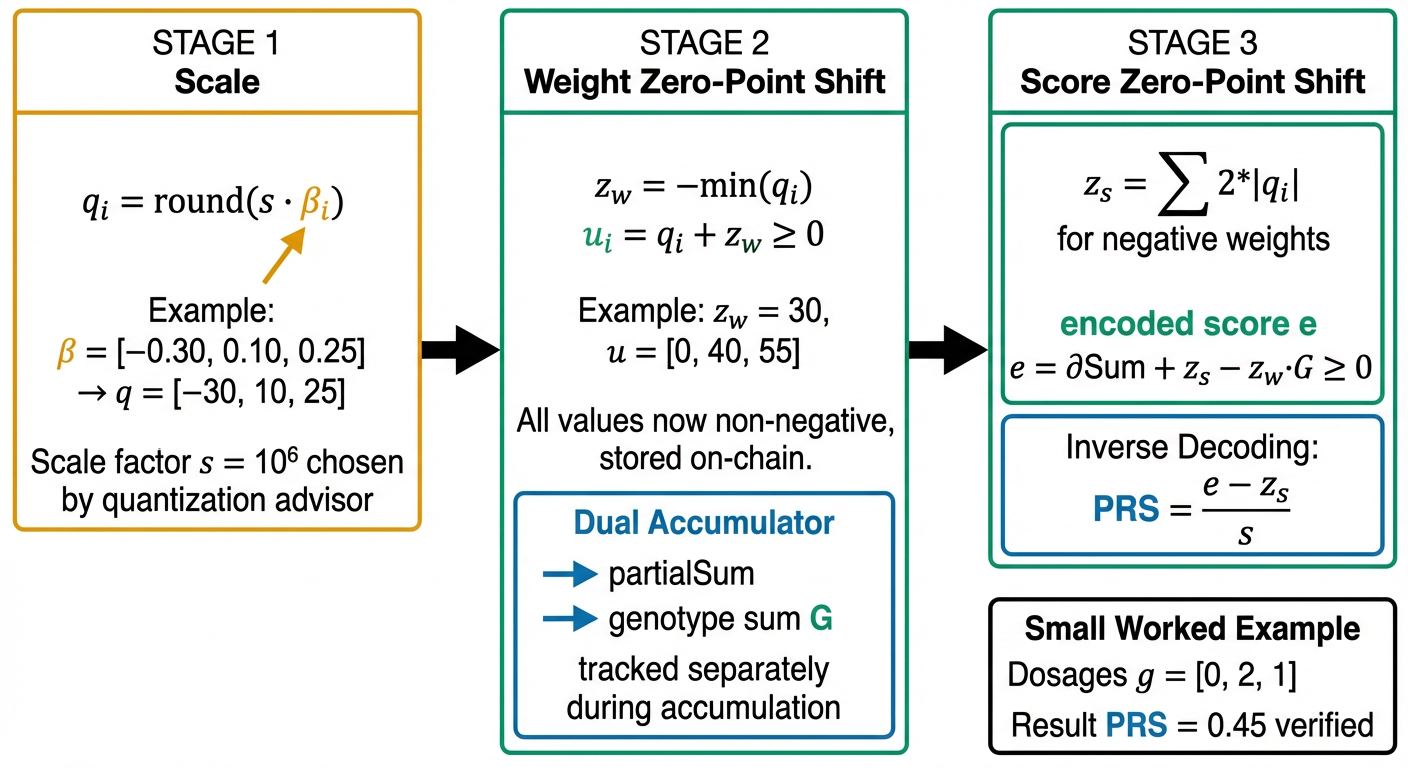}
  \caption{Three-step fixed-point quantisation scheme.  Signed GWAS floats are
  (1) scaled by $s$, (2) shifted by the weight zero-point $z_w$ to guarantee
  non-negative weights, and (3) shifted by the score zero-point $z_s$ to guarantee
  a non-negative encoded score.  Decoding inverts all three steps after decryption.}
  \label{fig:quantization}
\end{figure}

\textbf{Step 1: Scale.}
\begin{equation}
q_i = \mathrm{round}(s \cdot \beta_i), \quad q_i \in \mathbb{Z},
\end{equation}
where $s$ is a model-specific scale factor selected by an automated quantisation
advisor (Section~\ref{sec:advisor}).

\textbf{Step 2: Weight zero-point shift.}  Let
\begin{equation}
z_w = -\min_i q_i, \quad u_i = q_i + z_w \;\geq\; 0 \;\;\forall i.
\label{eq:wshift}
\end{equation}
The shifted weights $u_i \geq 0$ are stored on-chain.  The raw accumulation
$\sum_i g_i u_i$ overestimates the true PRS by the constant $z_w \cdot \sum_i g_i$,
requiring correction.  The engine tracks the genotype sum
$G = \sum_i g_i$ as a second encrypted accumulator to enable this correction.

\textbf{Step 3: Score zero-point shift.}
Even after the weight correction, the corrected sum
$P = \sum_i g_i u_i - z_w G$
can be negative for patients with many risk-decreasing alleles.  We define:
\begin{align}
z_s &= -\sum_{i:\beta_i<0} 2 q_i \;=\; \sum_{i:\beta_i<0} 2|q_i|,\\
e   &= P + z_s \;\geq\; 0.
\end{align}
$e$ is the \emph{encoded score} stored encrypted on-chain.  The on-chain
computation is structured to avoid any intermediate negative value:
\begin{equation}
e = \bigl(\mathrm{partialSum} + z_s\bigr) - \bigl(z_w \cdot G\bigr),
\label{eq:encode}
\end{equation}
where $z_s$ is added before subtracting $z_w G$, ensuring the intermediate
value $(\mathrm{partialSum} + z_s)$ is always non-negative.

\textbf{Decoding.}  After decryption by the authorised requester:
\begin{equation}
\mathrm{PRS} = \frac{e - z_s}{s}.
\end{equation}

\subsection{Worked Example}

Consider 3 genetic variants with weights $\boldsymbol{\beta} = [-0.30, 0.10, 0.25]$, scale
$s = 100$, and dosages $\mathbf{g} = [0, 2, 1]$.

\begin{enumerate}
  \item Quantise: $\mathbf{q} = [-30, 10, 25]$.
  \item Weight shift: $z_w = 30$, $\mathbf{u} = [0, 40, 55]$.
  \item Accumulate: partialSum $= 0{\cdot}0 + 2{\cdot}40 + 1{\cdot}55 = 135$;
        $G = 0 + 2 + 1 = 3$.
  \item Correction: $P = 135 - 30{\cdot}3 = 45$.
  \item Score shift: $z_s = 2{\cdot}30 = 60$; $e = 45 + 60 = 105$.
  \item Decode: PRS $= (105 - 60) / 100 = 0.45$.
  \item Verify: plaintext PRS $= 0{\cdot}(-0.30) + 2{\cdot}0.10 + 1{\cdot}0.25 = 0.45$. \checkmark
\end{enumerate}

\subsection{Overflow Safety}
\label{sec:overflow}

\begin{proposition}
Let $N$ be an arbitrary SNP count, let $s$ be the fixed-point scale factor, and
let $M$ satisfy $|q_i| \leq sM$ for every quantised weight $q_i$.  Then a
sufficient worst-case condition for unsigned 64-bit safety of the full encoding
pipeline is
\[
4 s M N \leq 2^{64} - 1.
\]
Equivalently,
\[
N \leq \left\lfloor \frac{2^{64} - 1}{4 s M} \right\rfloor.
\]
Under this condition, both the final encoded score $e$ and the intermediate
quantity $\mathrm{withOffset} = \mathrm{partialSum} + z_s$ in
Eq.~\ref{eq:encode} fit in $\texttt{uint64}$.
\end{proposition}

{\renewcommand{\qedsymbol}{}
\begin{proof}
Partition the quantised weights into positive, negative, and zero sets with
cardinalities $N_{+}$, $N_{-}$, and $N_{0}$, so that $N_{+}+N_{-}+N_{0}=N$.
Because $g_i \in \{0,1,2\}$ and $z_w=-\min_i q_i \leq sM$, the shifted weight
$u_i = q_i + z_w$ obeys
\[
u_i \leq
\begin{cases}
2 s M, & q_i > 0,\\
s M, & q_i \leq 0.
\end{cases}
\]
Hence the encrypted accumulation before correction satisfies
\[
\mathrm{partialSum} = \sum_i g_i u_i
\leq 4 s M N_{+} + 2 s M (N_{-}+N_{0}).
\]
The score zero-point satisfies
\[
z_s = -\sum_{i:q_i<0} 2 q_i \leq 2 s M N_{-}.
\]
Therefore the intermediate value materialised in Eq.~\ref{eq:encode} is bounded by
\[
\mathrm{withOffset} = \mathrm{partialSum} + z_s
\leq 4 s M N_{+} + 4 s M N_{-} + 2 s M N_{0}
\leq 4 s M N.
\]
This is the binding worst-case accumulator bound.

For the final encoded score, observe that the weight-shift correction exactly
cancels the zero-point contribution:
\[
P = \mathrm{partialSum} - z_w G = \sum_i g_i q_i.
\]
Thus
\[
-2 s M N_{-} \leq P \leq 2 s M N_{+},
\]
and after adding $z_s$ we obtain
\[
0 \leq e = P + z_s \leq 2 s M (N_{+}+N_{-}) \leq 2 s M N.
\]
So the encoded score range itself fits in $\texttt{uint64}$ under the weaker bound
$2 s M N \leq 2^{64}-1$, but the actual implementation must also accommodate the
larger intermediate $\mathrm{withOffset}$.  Therefore the sufficient adversarial
safety condition for all intermediate and final values is
\[
4 s M N \leq 2^{64}-1.
\]
For example, at the balanced scale $s=10^{6}$ with $M=1$, this yields the
mathematical worst-case ceiling
\[
N \leq \left\lfloor \frac{2^{64}-1}{4\times 10^6} \right\rfloor
\approx 4.61 \times 10^{12},
\]
which is many orders of magnitude above the fixture sizes used in our prototype
evaluation.
\end{proof}}

\subsection{Quantisation Advisor}
\label{sec:advisor}

Before publishing any model, an automated quantisation advisor evaluates candidate
scale values $s \in \{10^2, 10^4, 10^6, 10^8, 10^{10}\}$ against the actual GWAS
weight distribution.  For each $s$, the advisor computes the mean absolute error
(MAE) between the quantised dot product and the floating-point PRS for all
individuals in the fixture.  It outputs a recommendation from three tiers:

\begin{itemize}
  \item \textbf{Baseline} ($s \approx 10^2$): 1--15\% MAE. Proof-of-concept only.
  \item \textbf{Balanced} ($s \approx 10^6$): machine-epsilon error.
        Recommended default for all production models.
  \item \textbf{Max precision} ($s \approx 10^8$): no improvement over balanced
        on real GWAS fixtures; the limiting factor is source data precision.
\end{itemize}

The advisor runs in $\approx$200~ms and should be executed before every model
publication.  Gas cost is unaffected by scale choice---the bottleneck is SNP upload
transaction count, not arithmetic precision.

\section{Execution Protocols}
\label{sec:execution}

bioETH-PRS provides two computation strategies that differ in their gas cost,
storage semantics, and support for multi-party participation
(Figure~\ref{fig:protocol}).  Both strategies produce bit-identical results.

\begin{figure}[t]
  \centering
  \includegraphics[width=\columnwidth, keepaspectratio]{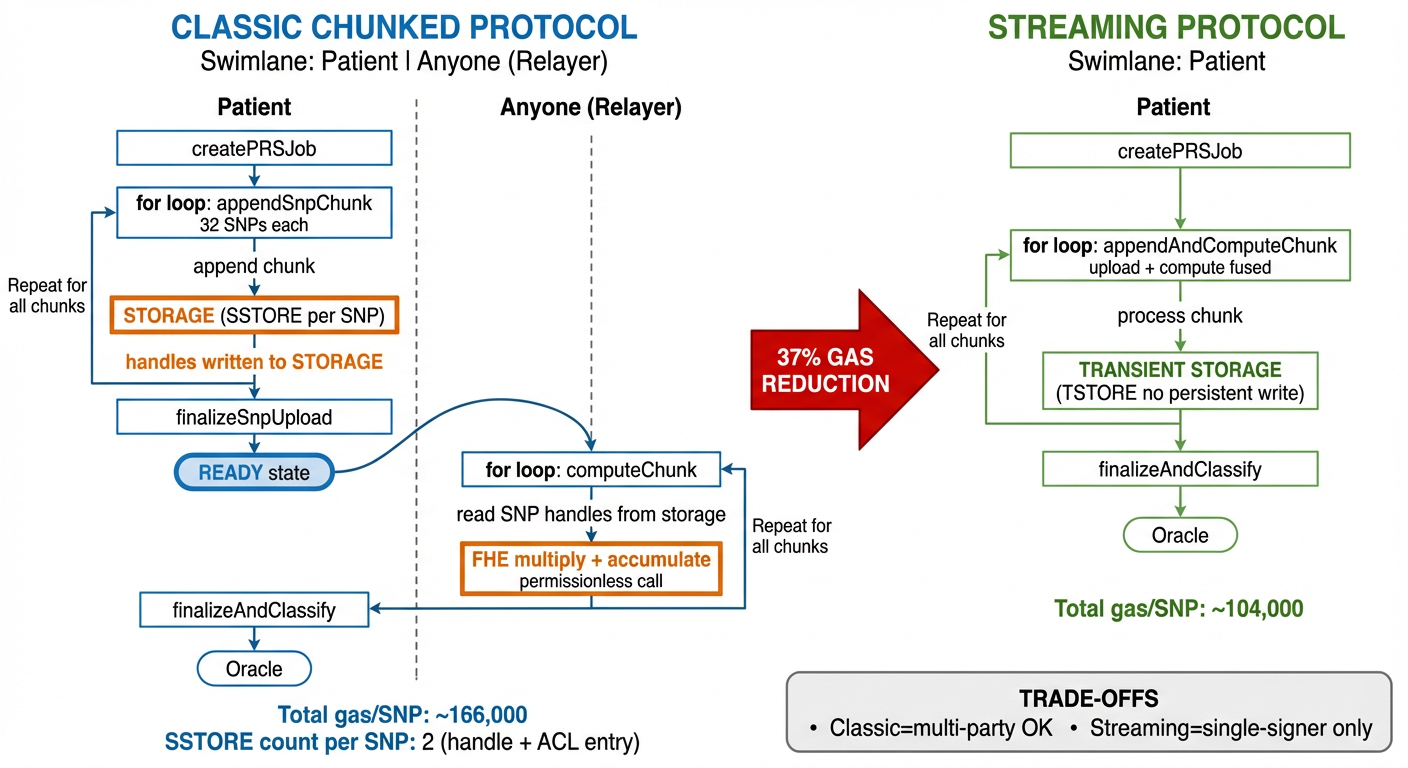}
  \caption{Execution protocol comparison.  The classic chunked path persists SNP
  handles in contract storage (SSTORE), enabling multi-party upload/compute
  separation.  The streaming path uses EIP-1153 transient storage (TSTORE),
  eliminating two SSTORE operations per SNP and reducing total gas by 37\%.  Both
  paths finalize through the Result Oracle.}
  \label{fig:protocol}
\end{figure}

\subsection{Classic Chunked Protocol}

The classic protocol separates SNP ingestion from computation, enabling different
parties to execute different phases at different times.

\begin{algorithm}[t]
\caption{Classic chunked PRS computation (patient and relayer)}
\label{alg:classic}
\small
\begin{algorithmic}[1]
\State \textbf{Patient:} \Call{createPRSJob}{modelId, sampleId}
\Comment{ACL checked; model geometry loaded}
\For{$k \gets 0$ to $\lceil N / r_u \rceil - 1$}
  \State \textbf{Patient:} \Call{appendSnpChunk}{jobId, $\mathrm{encSNPs}_{k}$, proof}
  \Comment{$\leq r_u = 32$ SNPs; handles persisted}
\EndFor
\State \textbf{Patient:} \Call{finalizeSnpUpload}{jobId}
\Comment{Job transitions to READY}
\For{$k \gets 0$ to $\lceil N / r_c \rceil - 1$}
  \State \textbf{Anyone:} \Call{computeChunk}{jobId}
  \Comment{Permissionless; reads persisted SNP handles}
  \State \quad $\mathrm{partialSum} \mathrel{+}= \sum_{i \in \mathrm{chunk}} g_i \cdot u_i$
  \State \quad $G \mathrel{+}= \sum_{i \in \mathrm{chunk}} g_i$
\EndFor
\State \textbf{Patient:} \Call{finalizeAndClassify}{jobId, oracle, $\tau_L$, $\tau_H$}
\end{algorithmic}
\end{algorithm}

The upload chunk size $r_u \leq 32$ is constrained by the fhEVM input-proof budget
(2048 bits / 64 bits per encrypted integer).  The compute chunk size $r_c \leq 20$
is constrained by the HCU budget: each compute chunk of $r_c$ SNPs executes
$3r_c + 2$ FHE operations (multiply, accumulate partial sum, accumulate genotype sum
per SNP, plus two coprocessor calls for accumulator handles).  At $r_c = 20$ this
yields 62 operations---within the measured HCU budget.

SNP ciphertext handles are persisted in a flat per-job mapping across transactions,
enabling upload and compute to be executed by different parties or at different
times.  The \textbf{compute phase is permissionless}: any third party may advance
computation by calling the compute function.  Because computation is deterministic
and non-interactive, a relayer learns no plaintext information and cannot affect
correctness or confidentiality.

\subsection{Streaming Protocol}

\begin{algorithm}[t]
\caption{Streaming PRS computation (patient, single-signer flow)}
\label{alg:streaming}
\small
\begin{algorithmic}[1]
\State \textbf{Patient:} \Call{createPRSJob}{modelId, sampleId}
\For{$k \gets 0$ to $\lceil N / r_c \rceil - 1$}
  \State \textbf{Patient:} append chunk $\mathrm{encSNPs}_{k}$ with proof
  \Comment{$r_c$ SNPs; handles transient only}
  \State \quad $\mathrm{partialSum} \mathrel{+}= \sum_{i \in \mathrm{chunk}} g_i \cdot u_i$
  \State \quad $G \mathrel{+}= \sum_{i \in \mathrm{chunk}} g_i$
  \Comment{SNP handles discarded after each call}
\EndFor
\State \textbf{Patient:} \Call{finalizeAndClassify}{jobId, oracle, $\tau_L$, $\tau_H$}
\end{algorithmic}
\end{algorithm}

The streaming protocol fuses upload and computation into a single call per chunk
(Algorithm~\ref{alg:streaming}).  Each call validates the incoming SNP handles
using EIP-1153 transient storage (scoped to the current transaction), immediately
multiplies them against model weights, accumulates the results, and discards the
handles.  No per-SNP persistent storage writes are performed.

\textbf{Gas savings.}
The classic path requires two persistent storage writes per SNP: a handle record in
contract storage and an ACL entry---each an Ethereum SSTORE at approximately
$25{,}000$ gas.  The streaming path eliminates both, saving $\approx$50,000 gas per
SNP, which accounts for the observed 37\% total reduction.  The irreducible cost
floor of $\approx$95,000--104,000 gas/SNP comprises the input-proof verification,
FHE multiply, and FHE accumulate---coprocessor-level costs not reducible by
contract design.

\textbf{Trade-off.}
The streaming path is incompatible with multi-party upload/compute separation:
upload and compute are coupled within the same transaction, restricting execution
to a single signer.  The classic path remains the appropriate choice for relay
services and architectures where upload and compute are performed by different
parties.

\section{Security Model}
\label{sec:security}

\subsection{Threat Model}

We consider a computationally bounded adversary with full network access to the
blockchain and coprocessor (Figure~\ref{fig:security}).  The adversary may: observe
all on-chain transactions and storage slots; operate a validator node and read all
encrypted handles; submit arbitrary transactions; and query classification outputs
repeatedly with chosen SNP inputs.  The adversary may not break the TFHE hardness
assumption (grounded in the RLWE problem) or forge ACL entries.

\begin{figure}[t]
  \centering
  \includegraphics[width=\columnwidth, keepaspectratio]{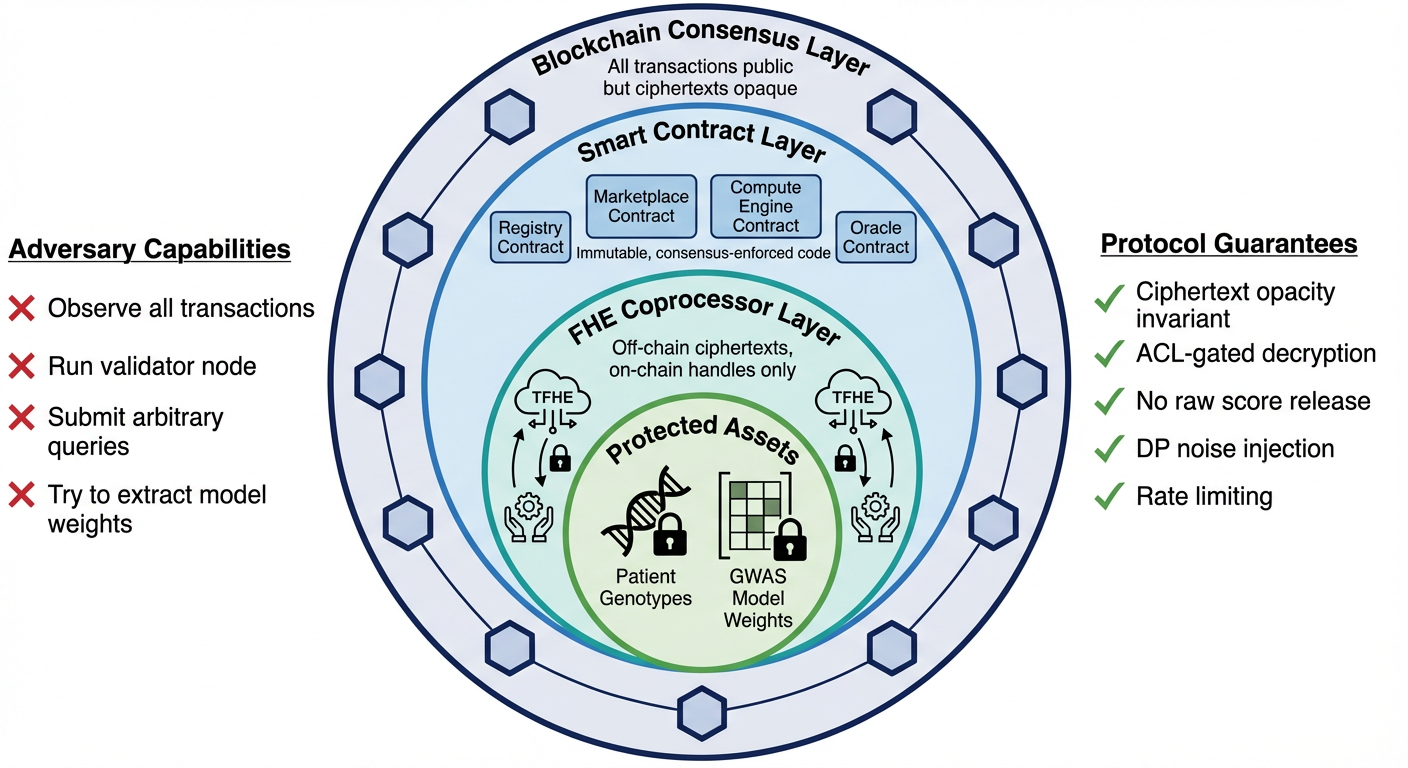}
  \caption{Security threat model and defence layers.  The TFHE/RLWE hardness
  assumption grounds the innermost layer.  Successive layers---FHE coprocessor,
  smart contracts, blockchain consensus---add verifiable security properties.
  Adversary capabilities and corresponding protocol guarantees are enumerated at
  the threat-boundary interface.}
  \label{fig:security}
\end{figure}

\subsection{Core Privacy Invariants}

We formalise the protocol's privacy guarantees as five critical invariants:

\begin{invariant}[Ciphertext opacity]
Encrypted handles are 32-byte opaque identifiers; raw ciphertexts reside
exclusively at the off-chain coprocessor.  No on-chain party stores or observes
plaintexts at any point during computation.
\end{invariant}

\begin{invariant}[ACL-gated decryption]
Every encrypted output is associated with an ACL record.  Decryption requires an
explicit ACL grant that can only be issued by the executing smart contract upon
satisfying the protocol's state-machine preconditions.
\end{invariant}

\begin{invariant}[No raw score release]
The final encoded PRS $e$ (Eq.~\ref{eq:encode}) is never made publicly
decryptable.  When oracle-required mode is disabled, the requester may obtain
ACL-gated decryption access to the raw score through \texttt{finalize()}.  When
oracle-required mode is enabled, that requester path is disabled: the raw score
handle is routed only to the approved oracle, while the Result Oracle emits an
encrypted noisy score handle and a publicly decryptable ternary category after
noise addition.
\end{invariant}

\begin{invariant}[Oracle-required mode]
A model owner may set an oracle-required flag, causing the compute engine to revert
on any attempt to release the raw score handle directly to the requester without
passing through the noisy output oracle.  This prevents requesters from bypassing
noise by decrypting raw scores.  When oracle-required mode is active, only the
registered approved oracle address may receive the encrypted raw-score handle.
\end{invariant}

\begin{invariant}[Single-finalize guarantee]
Each compute job carries a finalized flag.  Any finalization path sets this flag
and reverts on subsequent calls, preventing a requester from obtaining multiple
score handles for a single rate-limit slot.
\end{invariant}

\subsection{Noisy Output Release}

The current Result Oracle implements a DP-inspired noisy output release mechanism,
not a formally calibrated $(\varepsilon, \delta)$-differential privacy guarantee.
For each classification request, noise is drawn from $[0, B)$ where $B$ is a
power-of-two immutable constant set at oracle deployment:
\begin{equation}
e_{\mathrm{noisy}} = e + \nu, \quad \nu \sim \mathrm{Uniform}(0, B).
\end{equation}
The noise $\nu$ is generated by the fhEVM random source, unseeded by the caller
and determined solely by block randomness.  The caller has no influence over $\nu$.

All comparisons against classification thresholds operate on the noisy encrypted
score using encrypted boolean operations, avoiding any branching on hidden values:
\begin{align}
\hat{L} &= \mathbb{1}[e_{\mathrm{noisy}} < \tau_L], \notag\\
\hat{H} &= \mathbb{1}[e_{\mathrm{noisy}} \geq \tau_H], \notag\\
\mathrm{cat} &= \mathrm{select}\bigl(\hat{L},\;L,\;\mathrm{select}(\neg\hat{L}\wedge\neg\hat{H},\;M,\;H)\bigr).
\label{eq:classify}
\end{align}
The encrypted select primitive replaces conditional branching,
preventing side-channel leakage of the comparison result.  The minimum threshold-gap
requirement $\tau_H - \tau_L \geq B$ ensures the noisy score cannot deterministically
map to a single category regardless of $e$, reducing straightforward threshold-probing
strategies.

\textbf{Interpretation and limits.}
This mechanism is best understood as an anti-probing control that randomises the
released classification boundary.  A formal DP claim would require an explicit
neighbouring-dataset definition, a sensitivity analysis for the PRS score under the
chosen adjacency relation, and calibration of $B$ to concrete privacy parameters.
Those steps are not yet part of the present prototype, so we avoid assigning
specific $(\varepsilon, \delta)$ values.

\textbf{Bias correction.}
Uniform noise on $[0, B)$ introduces an expected upward bias of $B/2$.  Callers
must adjust classification thresholds by this amount, a deterministic correctable
offset exposed by the oracle's view function.  The anti-probing effect derives from
noise variance, not its mean.

\subsection{Anti-Probing: Rate Limiting}

Without query limits, an adversary could mount a model-extraction attack by
submitting many classification queries with crafted SNP inputs and observing
the resulting categories.  The compute engine enforces per-model, per-wallet, and
per-sample block-windowed job quotas to raise this cost.

Let $R$ be the maximum queries per $W$-block window, $B$ the noise bound, and
$K = 3$ the number of output categories.  Each query reveals at most
$\log_2(K) \approx 1.58$ bits minus noise entropy.  At suggested settings for
private models ($R = 3$, $W = 1000$, $B = 128$), extracting a single 20-bit weight
requires approximately $2 \times 10^4 / (3 \times 1.58) \approx 4{,}220$ block
windows, corresponding to $\approx$2,800 hours at 12~s/block.  Block-based windows
(rather than timestamps) prevent miner manipulation of window boundaries.

\section{Empirical Evaluation}
\label{sec:evaluation}

\subsection{Experimental Setup}

\textbf{Prototype evaluation (mock coprocessor).} We evaluate bioETH-PRS against the four fixture sizes (100, 500, 1,000, and 5,000
SNPs) from the HEPRS reference dataset \citep{knight2026heprs}, which contains real
GWAS beta weights for a schizophrenia model at each size.  Gas measurements use a
Hardhat in-process mock coprocessor that validates the complete fhEVM protocol
(handles, ACL, input proofs) while performing plaintext arithmetic; gas numbers
are expected to be within 10--20\% of real-network deployment.
The HCU ceiling was measured empirically: chunk sizes of 20~SNPs (62 ops) pass;
25~SNPs (77 ops) fail, confirming the measured 60--74-operation HCU budget.
All experiments use
upload chunk size $r_u = 32$ and compute chunk size $r_c = 20$.

\subsection{Quantisation Accuracy}

At the balanced scale tier ($s \approx 10^6$), reconstructed PRSs agree with
floating-point computation to machine epsilon across all 50 individuals in each of
the four fixtures (200 individual-fixture combinations), as shown in
Table~\ref{tab:quantisation}.  No negative encoded scores or unsigned 64-bit integer
overflow events were observed.  The baseline tier ($s \approx 10^2$) yielded
1--15\% MAE---unsuitable for clinical use, as a 15\% score error is equivalent
to misclassifying borderline patients.  The max-precision tier ($s \approx 10^8$)
provided no additional accuracy improvement: the limiting factor is the number of
significant digits in the source GWAS betas rather than the encoding resolution.
This confirms that the balanced tier ($s \approx 10^6$) is both necessary and
sufficient for production deployment.

\begin{table}[ht]
\centering
\caption{Quantisation accuracy on HEPRS fixtures (balanced tier)}
\label{tab:quantisation}
\small
\begin{tabular}{rrrll}
\toprule
\textbf{SNPs} & \textbf{Scale $s$} & \textbf{Bits req.} & \textbf{MAE} & \textbf{Indiv.} \\
\midrule
100   & $3\times10^6$ & 16 & $<\epsilon_\mathrm{mach}$ & 50/50 pass \\
500   & $3\times10^6$ & 16 & $<\epsilon_\mathrm{mach}$ & 50/50 pass \\
1,000 & $1\times10^6$ & 16 & $<\epsilon_\mathrm{mach}$ & 50/50 pass \\
5,000 & $1\times10^6$ & 16 & $<\epsilon_\mathrm{mach}$ & 50/50 pass \\
\bottomrule
\end{tabular}
\end{table}

\subsection{Gas Consumption and Scaling}

Figure~\ref{fig:gas} and Table~\ref{tab:gas} report total gas consumption across
execution paths and fixture sizes in the mock-coprocessor prototype environment.
Gas scales linearly with SNP count in both paths, confirming the absence of hidden
quadratic overhead at the contract layer.  The streaming path saves approximately
37\% above 500 SNPs, driven entirely by the elimination of two SSTORE operations per SNP.

\begin{figure}[t]
  \centering
  \includegraphics[width=\columnwidth, keepaspectratio]{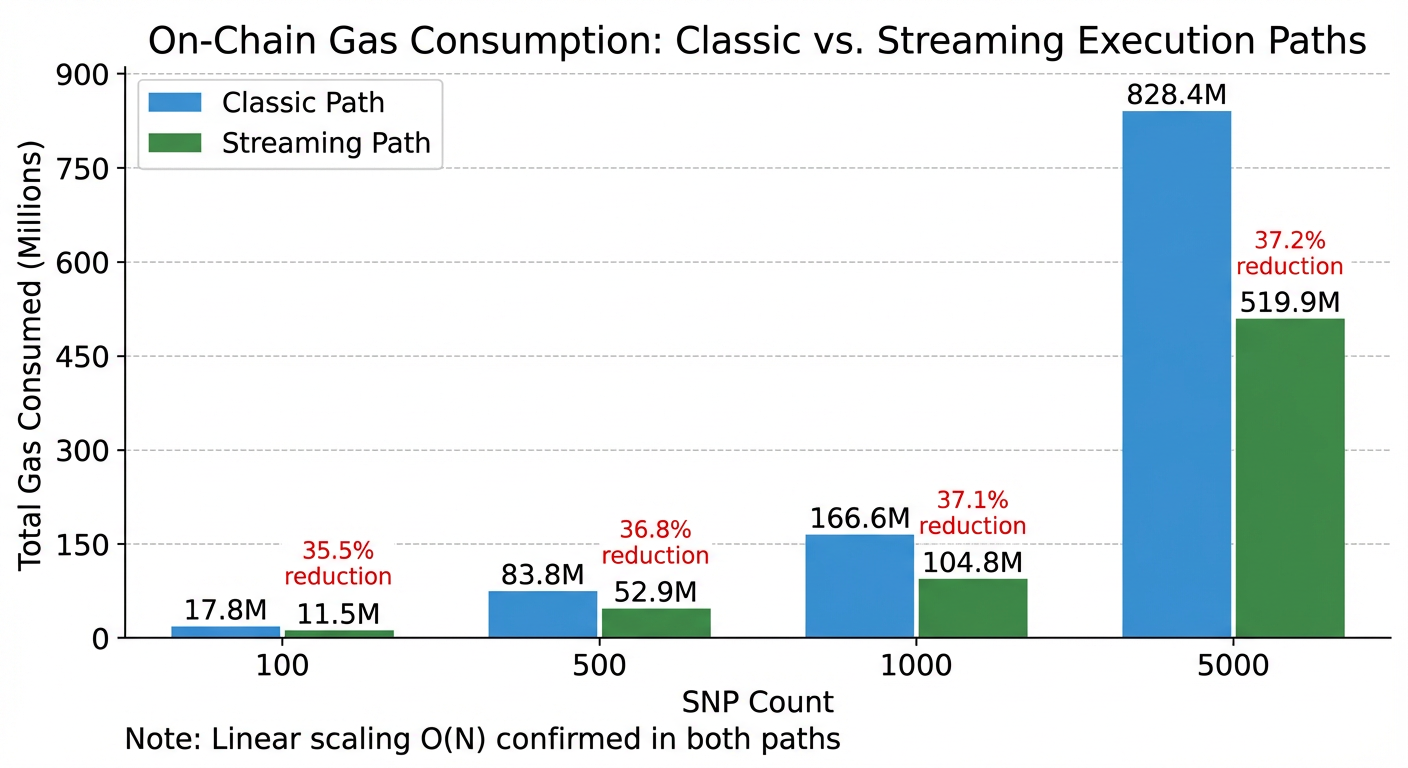}
  \caption{On-chain gas consumption versus SNP count for both execution paths.
  Gas scales linearly ($\mathcal{O}(N)$) in both cases.  The streaming path
  achieves 35.5--37.2\% gas reduction over the classic chunked path by eliminating
  persistent SSTORE operations for intermediate SNP ciphertext handles.}
  \label{fig:gas}
\end{figure}

\begin{table}[ht]
\centering
\caption{Mock-environment total gas consumption by execution path}
\label{tab:gas}
\small
\begin{tabular}{rrrr}
\toprule
\textbf{SNPs} & \textbf{Classic} & \textbf{Streaming} & \textbf{Reduction} \\
\midrule
100   & 17,822,364  & 11,496,911  & 35.5\% \\
500   & 83,943,764  & 53,021,675  & 36.8\% \\
1,000 & 166,832,473 & 104,972,012 & 37.1\% \\
5,000 & 829,373,084 & 520,440,311 & 37.2\% \\
\bottomrule
\end{tabular}
\end{table}

\subsection{Per-SNP Cost Decomposition}

Table~\ref{tab:persnp} decomposes per-SNP gas.  The irreducible floor of
$\approx$95,000--104,000 gas/SNP comprises the input-proof verification, FHE
multiply, and FHE accumulate---operations set by the fhEVM protocol and not
reducible by contract design choices.  Future improvements in coprocessor
throughput or lighter-weight proof schemes would directly reduce this floor.

\begin{table}[ht]
\centering
\caption{Per-SNP gas decomposition}
\label{tab:persnp}
\small
\setlength{\tabcolsep}{4pt}
\begin{tabularx}{\columnwidth}{@{}Yrr@{}}
\toprule
\textbf{Operation} & \textbf{Classic} & \textbf{Streaming} \\
\midrule
Input proof verification         & $\sim$50,000 & $\sim$50,000 \\
FHE multiply                     & $\sim$27,000 & $\sim$27,000 \\
FHE accumulate (partial sum)     & $\sim$27,000 & $\sim$27,000 \\
Handle SSTORE (classic only)     &  $\sim$25,000 & --- \\
ACL SSTORE (classic only)        &  $\sim$25,000 & --- \\
SLOAD, misc.\ overhead           &   $\sim$12,000 & --- \\
\midrule
\textbf{Total/SNP}               & $\approx$\textbf{166,000} & $\approx$\textbf{104,000} \\
\bottomrule
\end{tabularx}
\end{table}

\subsection{Latency}

On the Hardhat mock coprocessor, end-to-end wall-clock times for the streaming
path are: 100 SNPs: $\approx$386~ms; 500 SNPs: $\approx$1.5~s; 1,000 SNPs:
$\approx$2.8~s; 5,000 SNPs: $\approx$14.5~s.  These are local simulation
timings; real-network latency is dominated by per-block confirmation time, as each
chunk corresponds to one on-chain transaction.  The transaction count scales as
$\lceil N / r_c \rceil + 2$ for the streaming path.

\subsection{Deployment Cost Projections}

Table~\ref{tab:cost} projects per-analysis cost across three gas pricing scenarios.
These figures are scenario analyses derived from mock-environment gas measurements,
not measured deployment prices.  Ethereum L1 pricing ($30$ gwei) renders the system
economically unviable for clinical use.  At L2-equivalent pricing on a purpose-built
FHE rollup ($0.05$ gwei), 100-SNP and 500-SNP analyses project to \$1.72 and \$7.95,
respectively.  Under such low-gas assumptions, the system may become cost-competitive
with centralised commercial genomics services (\$50--\$300), but that comparison
remains contingent on real-network deployment.  Dedicated application-chain deployment
reduces costs by two further orders of magnitude in the same pricing model, making even
5,000-SNP analyses economically plausible at \$1.56 per run.

\begin{table}[ht]
\centering
\caption{Projected USD cost per analysis (streaming path, ETH at \$3,000)}
\label{tab:cost}
\footnotesize
\setlength{\tabcolsep}{3pt}
\begin{tabular}{rrrr}
\toprule
\textbf{SNPs} & \textbf{L1 (30~gwei)} & \textbf{L2 (0.05~gwei)} & \textbf{App\-chain} \\
\midrule
100   & \$1,035  & \$1.72  & \$0.034 \\
500   & \$4,772  & \$7.95  & \$0.159 \\
1,000 & \$9,447  & \$15.75 & \$0.315 \\
5,000 & \$46,840 & \$78.07 & \$1.561 \\
\bottomrule
\end{tabular}
\end{table}

\subsection{Correctness and Protocol Verification}

For each fixture size, we verified that the on-chain encoded score, after
decryption and application of the quantisation inverse mapping, matches the
plaintext floating-point dot product.  For the 100-SNP HEPRS fixture, the expected
encoded score of 758,685 was reproduced exactly across all validation runs.  Both
execution paths produce identical scores for all four fixture sizes (verified with
plaintext arithmetic in the mock coprocessor).

Protocol invariants---ACL enforcement, state-machine integrity, single-finalize
guarantee, minimum threshold gap, oracle-required mode, approved-oracle
enforcement, rate-limiting window behaviour, and ACL revocation handling---were
validated across 134 automated test cases.  The HCU ceiling probe confirmed the
mock budget at 60--74 operations per transaction.

\section{Access Control and Compute Flows}
\label{sec:acl}

\subsection{Encrypted Handle Lifecycle}

Every encrypted value in the fhEVM has an associated ACL record maintained in a
dedicated on-chain contract.  The ACL maps handles to authorised addresses and
distinguishes four grant types with different persistence semantics:

\begin{itemize}
  \item \textbf{Persistent contract grant}: Written to the ACL contract's persistent
        storage via SSTORE.  Allows the owning contract to use the handle in future
        transactions.  This is the primary source of gas cost in the classic upload
        path ($\approx$25,000 gas per SNP handle).
  \item \textbf{Persistent user grant}: Written to the ACL contract via SSTORE.
        Allows an end-user to decrypt the handle via the KMS gateway.  Applied at
        finalization to grant the requester access to their encoded score.
  \item \textbf{Transient grant}: Uses EIP-1153 transient storage.  Valid only for
        the current transaction.  The streaming path relies exclusively on transient
        grants for intermediate SNP handles, avoiding the SSTORE cost entirely.
  \item \textbf{Public grant}: Enables anyone to decrypt.  Applied only to the
        ternary risk category after noise addition; never applied to score values.
\end{itemize}

The ACL design is the root of the 37\% gas differential between the classic and
streaming paths.  In the classic path, each of the $N$ SNP handles requires a
persistent contract grant (SSTORE in the ACL) and a handle write to contract
storage (second SSTORE).  In the streaming path, transient grants are used for
SNP handles and the SSTORE costs are eliminated.

\subsection{State Machine and Mutual Exclusion}

The PRS Compute Engine enforces a strict job state machine to prevent incorrect
ordering of operations:
\begin{equation*}
\texttt{PENDING} \!\to\! \texttt{UPLOADING}
\!\to\! \texttt{READY}
\!\to\! \texttt{COMPUTING} \!\to\! \texttt{DONE}.
\end{equation*}
The classic and streaming paths are mutually exclusive per job.  The guard condition
uses upload counts: if any SNP handles have been written to persistent storage, the
classic path is in use and streaming calls are rejected; if computation has advanced
without persistent uploads, the streaming path is in use.  This prevents hybrid
calls that could compromise the handle-lifecycle invariants.

\subsection{Private Model Access Control}

When a GWAS model is published with private (encrypted) weights, access control is
layered: the model owner must first authorise the compute engine as a reader; individual
requesters must also be added to the private model reader list; and each compute chunk
re-checks reader authorisation, so revocation of a requester mid-job causes the next
compute chunk to fail, effectively terminating the job.  This is in contrast to the
registry ACL, which is checked only at job creation.  The asymmetric behaviour between
registry and model ACL is documented and empirically verified.

\section{Discussion}
\label{sec:discussion}

Our work presents bioETH-PRS, a privacy-preserving framework for polygenic risk scoring that replaces the trusted third-party evaluator used in prior homomorphic-encryption approaches with auditable smart contracts on an FHE-enabled blockchain. It enables encrypted computation of PRS while protecting both patient genotypes and GWAS model weights, introduces an overflow-safe quantization scheme, and demonstrates linear scaling, exact score correctness, and a 37

\subsection{HEPRS and bioETH-PRS: Complementary Systems}

bioETH-PRS and HEPRS \citep{knight2026heprs} are best understood as complementary
rather than competing systems occupying different points in a design space defined
by trust requirements and scalability.

HEPRS is appropriate when: (1) three-party non-collusion can be organisationally
enforced (e.g., hospital evaluator, academic modeler, clinical client); (2) SNP
counts exceed 5,000; and (3) the evaluation environment provides sufficient RAM
($>$65~GB for 1,000 individuals at 110,000 SNPs).  In this context, CKKS's native
signed float arithmetic eliminates the need for quantisation.

bioETH-PRS is appropriate when: (1) no trusted evaluator can be assumed; (2)
model weights must remain private not just from the patient but from the evaluator
as well; (3) computation must be verifiable without re-running it; or (4) the use
case benefits from public, decentralised access without institutional infrastructure.
These conditions commonly arise in genomic research consortia, cross-institutional
data sharing, and commercial PRS services.

\subsection{Limitations and Open Problems}

\textbf{SNP count ceiling.}
In our mock-coprocessor measurements (Section~\ref{sec:evaluation}), the HCU budget
of 60--74 FHE operations per transaction limits each compute chunk to 20 SNPs.
Scaling
to genome-wide models ($>$50,000 SNPs) would require thousands of transactions,
making on-chain costs prohibitive at any current gas pricing.  HEPRS remains the
appropriate tool at that scale.  Future FHE coprocessor improvements may raise this
ceiling.

\textbf{SNP provenance.}
The system verifies ACL access to a registered sample but cannot verify that
submitted encrypted SNP values faithfully represent that sample's genotype.  A
malicious requester could submit arbitrary ciphertexts.  Rate limiting and noisy classification mitigate information leakage per query but do not prevent fabricated
queries.

\textbf{URI observability.}
Sample data URIs are on-chain plaintext.  Any network node can observe storage
slots directly, even if read access is ACL-gated at the application layer.  Storing
only a cryptographic commitment to the URI would address this residual exposure.

\textbf{DP bias.}
Uniform noise on $[0, B)$ introduces an expected upward bias of $B/2$.  True
zero-mean DP would require signed noise, which the current fhEVM does not natively
support.  Threshold adjustment is a correctable workaround but adds complexity for
deployers.

\textbf{Gas cost at L1 pricing.}
Commercial viability requires a purpose-built FHE rollup or application chain.
At Ethereum L1 gas prices, any SNP count is economically indefensible for routine
clinical use.

\subsection{Future Directions}

Planned extensions include: (1) LD-aware weighting in the encrypted domain,
incorporating LDpred2-style shrinkage \citep{prive2020ldpred2}; (2) multi-party key
generation enabling federated cohort analyses without a single decryption authority;
(3) model versioning and deprecation mechanisms; (4) SNP handle re-use across
multiple models from the same registered sample; and (5) token-based fee mechanisms
to economically incentivise model quality.  The architecture's modularity---four
independently upgradeable contracts---is designed to accommodate these extensions
without disrupting existing deployments.

\section{Related Work}
\label{sec:related}

Privacy-preserving genomic computation has been pursued through multiple
cryptographic paradigms.  Kim and Lauter \citep{kim2015} applied homomorphic
encryption schemes to compute GWAS statistics for up to 5,000 sequences, but did not
extend to PRS computation.  Blatt et al.\ \citep{blatt2020} built a CKKS-based
privacy-preserving GWAS pipeline for large-scale studies, noting that PRS
computation from decrypted summary statistics was future work.  McLaren et al.\
\citep{mclaren2016} demonstrated privacy-preserving genomic testing in an HIV
clinical-care model using homomorphic encryption, but did not address the full PRS
pipeline.  Raisaro et al.\ \citep{raisaro2019} introduced MedCo for secure,
privacy-preserving exploration of distributed clinical and genomic data.

Knight et al.\ \citep{knight2026heprs} (HEPRS) is the closest prior work, providing
the first complete FHE pipeline for PRS with real clinical data.  HEPRS achieves
$r > 0.999$ correlation with plaintext scores on a 110,000-SNP schizophrenia
model and demonstrates practical feasibility on commodity hardware.  bioETH-PRS
builds directly on the HEPRS computational framework and uses the same GWAS fixture
datasets, extending it to a blockchain setting that removes the designated evaluator.

On the blockchain side, prior work on privacy-preserving smart contracts has
focused on anonymous payments and related protocol work
\citep{sasson2014zerocash,pertsev2019tornado} and
zero-knowledge proof-based computation \citep{ben2018starkware}, which enable
verifiable computation but not confidential computation, ZK-SNARKs prove knowledge
without concealing the program inputs.  FHE on the blockchain, typified by the fhEVM
framework \citep{zamafhevm}, enables both confidentiality and verifiability in
a single system.  To our knowledge, bioETH-PRS is the first application of fhEVM
to clinical genomics.

\section{Conclusion}
\label{sec:conclusion}

We presented bioETH-PRS, a protocol for confidential polygenic risk
scoring using TFHE-based fully homomorphic encryption on a programmable blockchain.
By replacing the centralised evaluator of HEPRS \citep{knight2026heprs} with
immutable, consensus-enforced smart contracts, we remove the designated evaluator assumption that is the principal residual vulnerability of prior centralised FHE approaches to genomic computation, while still depending on the fhEVM stack, contract correctness, ACL/decryption infrastructure, and blockchain liveness.

Key contributions are: a four-layer contract architecture with independently
auditable trust boundaries; a provably overflow-safe three-step fixed-point
quantisation scheme achieving machine-epsilon reconstruction accuracy on real GWAS
datasets; a streaming execution path reducing on-chain gas consumption by 37\%; and
an on-chain noisy output oracle with cryptographically generated noise, mandatory
activation, and anti-probing rate limiting.

Prototype evaluation on HEPRS fixtures confirms linear gas scaling and exact score
correctness.  Under L2-equivalent gas pricing assumptions, projected costs range
from USD\,\$1.72 to \$7.95 for 100--500 SNP analyses.  That range suggests the
approach may be practical for curated clinical PRS panels in low-gas environments,
while offering stronger evaluator minimisation than centralised alternatives.

The system demonstrates that, for the 100--5,000 SNP range covering a subset of
curated polygenic models, on-chain FHE computation without a trusted
evaluator is technically feasible and may become economically practical on specialised
low-gas networks.  Continued improvements in FHE
coprocessor throughput and purpose-built chain gas pricing will expand this range
toward genome-wide models and bring private, verifiable genomic computation within
reach of routine clinical practice.

\section*{Biographical Note}
Kimon Antonios Provatas, Christos Galanopoulos, and Ilias Georgakopoulos-Soares are researchers at The University of Texas at Austin studying pharmacology, genomics, and privacy-preserving biomedical computation at the Dell Pediatric Research Institute.

\section*{Data Availability Statement}
No new datasets were generated or analysed in this study.

\section*{Code Availability Statement}
All relevant code and results can be found on GitHub at \url{https://github.com/Georgakopoulos-Soares-lab/bioETH-PRS}.

\bibliographystyle{abbrvnat}
\bibliography{bioeth_prs}

\end{document}